\documentclass[11pt,twoside]{article}

\usepackage{a4wide, amssymb, amsmath, amsthm, graphics, comment, xspace, enumerate}
\usepackage{graphicx,multirow}
\usepackage[a4paper,colorlinks=true,citecolor=blue,urlcolor=blue,linkcolor=blue,
            bookmarksopen=true,unicode=true,pdffitwindow=true]{hyperref}
\usepackage[section]{algorithm}
\usepackage[noend]{algpseudocode}
\usepackage{caption}

\hypersetup{pdftitle={Graph Operations and Irregularity2 of Graph}}
\usepackage{amsfonts,subfigure}
\usepackage{tikz}
\usetikzlibrary{arrows}
\usetikzlibrary{decorations}
\usetikzlibrary{patterns}
\subfiguretopcaptrue
\tikzstyle{vtx}=[circle, inner sep= 0pt, minimum size= 1.2mm, fill]   
\newtheorem{te}{Theorem}[section]

\newcommand{\beq}{\begin{eqnarray}}
\newcommand{\eeq}{\end{eqnarray}}

\newcommand{\beqs}{\begin{eqnarray*}}
\newcommand{\eeqs}{\end{eqnarray*}}
\newcommand{\irr}{{\rm irr}}
\newcommand{\imb}{{\rm imb}}

\allowdisplaybreaks

\begin{document}
\title{The Total Irregularity of Graphs under Graph Operations} 
\maketitle
\begin{center}
{\large \bf Hosam Abdo, Darko Dimitrov}
\end{center}
\baselineskip=0.20in
\begin{center}
{\it Institut f\"ur Informatik, Freie Universit\"{a}t Berlin, \\ 
Takustra{\ss}e 9, D--14195 Berlin, Germany} \\ 
E-mail: {\tt [abdo,darko]@mi.fu-berlin.de} \\
\end{center}
\vspace{5mm}
\begin{abstract}
The {\em total irregularity} of a graph $G$ is defined as $\irr_t(G)=\frac{1}{2} \sum_{u,v \in V(G)}$ $|d_G(u)-d_G(v)|$, 
where $d_G(u)$ denotes the degree of a vertex $u \in V(G)$. In this paper we give (sharp) upper bounds on the total
irregularity of graphs under several graph operations including join, lexicographic product, Cartesian product, 
strong product, direct product, corona product, disjunction and symmetric difference.
\end{abstract}
{\small \hspace{0.25cm}\textbf{Keywords:} 
Irregularity and total irregularity of graphs, graph operations}
%
%
\section[Introduction]{Introduction}

Let $G$ be a simple undirected graph with $|V(G)|=n$ vertices and $|E(G)|=m$ edges. The \emph{degree} of a vertex $v$ in $G$ 
is the number of edges incident with $v$ and it is denoted by $d_G(v)$. A graph $G$ is \emph{regular} if all its vertices have the same degree, 
otherwise it is \emph{irregular}. However, in many applications and problems it is of big importance to know how irregular a given graph is.
Several graph topological indices have been proposed for that purpose. Among the most investigated ones are: the irregularity of a graph 
introduced by Albertson \cite{Albertson_Irr}, the variance of vertex degrees \cite{Bell_1}, and Collatz-Sinogowitz index \cite{CollSin-57}.

The \emph{imbalance} of an edge $e=uv \in E$, defined as $\imb(e)=\left|d_G(u)-d_G(v)\right|$, appeares implicitly in the context 
of Ramsey problems with repeat degrees \cite{AlbertsonRamsey}, and later in the work of Chen, Erd\H{o}s, Rousseau, and Schlep \cite{CERS93}, 
where $2$-colorings of edges of a complete graph were considered. In \cite{Albertson_Irr}, Albertson defined the \emph{irregularity} of $G$ as
\beq \label{eqn:003}
\irr(G) = \sum_{e \in E(G)}\imb(e).
\eeq
It is shown in \cite{Albertson_Irr} that for a graph $G$, $\irr(G) < 4 n^3/27$ and that this bound can be approached arbitrary closely.
Albertson also presented upper bounds on irregularity for bipartite graphs, triangle-free graphs and a sharp upper bound for trees.
Some claims about bipartite graphs given in \cite{Albertson_Irr} have been formally proved in \cite{Henning-Rautenbach}.
Related to Albertson is the work of Hansen and M{\'e}lot \cite{Hansen_Melot}, who characterized the graphs with $n$
vertices and $m$ edges with maximal irregularity.
The irregularity measure $\irr$ also is related to the \emph{first Zagreb index} $M_1(G)$ and the \emph{second Zagreb index} $M_2(G)$,
one of the oldest and most investigated topological graph indices, defined as follows:
\beq
M_1(G) = \sum_{v\in V(G)}d^2_G(v)  \quad \mbox{and} \quad
M_2(G) = \sum_{uv\in E(G)}d_G(u)d_G(v). \nonumber
\eeq
Alternatively the first Zagreb index can be expressed as
\beq \label{eqn:001_2}
M_1(G) =\sum_{uv\in E(G)}\left[d_G(u)+d_G(v)\right]. 
\eeq
Fath-Tabar \cite{G2} established new bounds on the first and the second Zagreb indices that depend on the
irregularity of graphs as defined in (\ref{eqn:003}). In line with the standard terminology of chemical graph theory,
and the obvious connection with the first and the second Zagreb indices, Fath-Tabar named the sum in (\ref{eqn:003}) 
the \emph{third Zagreb index} and denoted it by $M_3(G)$. 
The graphs with maximal irregularity with $6$, $7$ and $8$ vertices are depicted in 
Figure \ref{KS-max-irr}.
\begin{figure}[htbp]
\centering
\vspace{-0.3cm}
\includegraphics[scale=1.0]{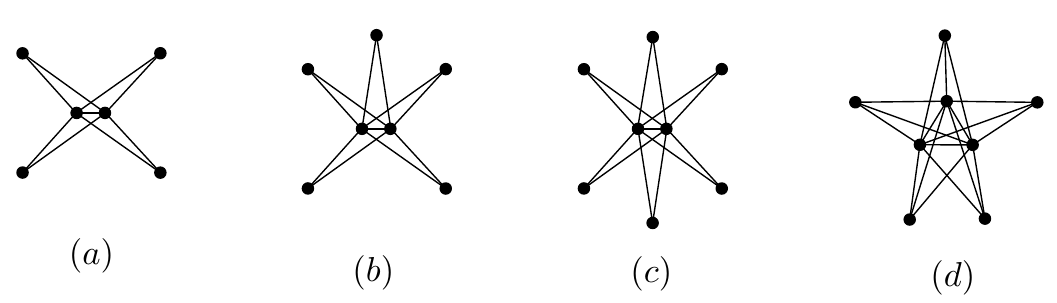} 
\caption{(a) The graph with $6$ vertices with  maximal $\irr$.
              (b)The graph with $7$ vertices with  maximal $\irr$. 
              (c) and (d) Graphs with $8$ vertices with  maximal $\irr$.}
\label{KS-max-irr}   
\vspace{-0.1cm}
\end{figure}

Two other most frequently used graph topological indices, that measure how irregular a graph is, are
the \emph{ variance of degrees} and the \emph{ Collatz-Sinogowitz index} \cite{CollSin-57}.  Let $G$ be a graph
with $n$ vertices and $m$ edges, and $\lambda_1$ be the index or largest eigenvalue of the adjacency matrix
$A = (a_{ij})$ (with $a_{ij} = 1$ if vertices $i$ and $j$ are joined by an edge and $0$ otherwise). Let $n_i$ denotes
the number of vertices of degree $i$ for $i =1, 2,\dots, n - 1$.
Then, the variance of degrees and the Collatz-Sinogowitz index are respectively defined as
\beq \label{eqn:001_3}
\mbox{Var(G)} =\frac{1}{n} \sum_{i=1}^{n-1} n_i \left( i - \frac{2 m}{n}\right)^2
\qquad \mbox{and} \qquad
\mbox{CS(G)} = \lambda_1 - \frac{2 m}{n}.
\eeq
Results of comparing $\irr$, $\mbox{CS}$ and $\mbox{Var}$  are presented in \cite{Bell_1, CvetRow88, GHM-05}.

There have been other attempts to determine how irregular graph is \cite{Alavi-88, Alavi-87, Alavi-Liu, Bell_2, Char-88, Char-87, JacEnt-86}, 
but heretofore this has not been captured by a single parameter as it was done by the irregularity measure by Albertson. 

The graph operation, especially graph products, plays significant role not only in pure and applied mathematics, but also in 
computer science. For example, the Cartesian product provide an important model for linking computers.
In order to synchronize the work of the whole system it is necessary to search for Hamiltonian paths
and cycles in the network. Thus, some results on Hamiltonian paths and cycles in Cartesian product of
graphs can be applied in computer network design \cite{YtpSg09}. Many of the problems can be easily handled if
the related graphs are regular or close to regular.

Recently in~\cite{Dimit-Abdo} a new measure of irregularity of a graph, so-called the {\it total irregularity},
that depends also on one single parameter (the pairwise difference of vertex degrees) was introduced. 
It was defined as
\beq \label{eqn:002}
\irr_t(G) = \frac{1}{2}\sum_{u, v\in V(G)} \left| d_G(u)-d_G(v) \right|.
\eeq
In the next theorem the upper bounds on the total irregularity of a graph are presented. Graphs with maximal 
total irregularity are  depicted in Figure \ref{max-total-irr}.
\begin{te} [\cite{Dimit-Abdo}] \label{thm-graphs-max-irr_t}
\label{irr2_bound}
For a simple undirected graph $G$ with $n$ vertices, it holds that 
\beq 
\irr_t(G) \leq 
    \begin{cases}  
			  \frac{1}{12}(2n^3 - 3n^2 - 2n)    &\mbox{n even,}\\
				\\
				\frac{1}{12}(2n^3 - 3n^2 - 2n +3)  &\mbox{n odd.} \nonumber
			\end{cases}
\eeq
Moreover, the bounds are sharp.
\end{te}
\begin{figure}[htbp]
\centering
\vspace{-0.1cm}
\begin{minipage}[h]{6cm}
\begin{tikzpicture}
        \path (0,0) coordinate (P0);
        \path (1,0) coordinate (P1);
        \path (3,0) coordinate (P2);
        \path (4,0) coordinate (P3);
        \path (0,-1.5) coordinate (Q0);
        \path (1,-1.5) coordinate (Q1);
        \path (3,-1.5) coordinate (Q2);
        \path (4,-1.5) coordinate (Q3);

        \foreach \x in {0,...,3} {
	\fill (P\x) circle (2pt);
	\fill (Q\x) circle (2pt);}
        \foreach \x in {0,...,3} {
             \draw[dashed] (P\x) -- (Q\x) ;}

\path   (2,0) node {$\ldots$};
\path   (2,-1.5) node {$\ldots$};

             \draw (P0) -- (Q1) ;
             \draw (P0) -- (Q2) ;
             \draw (P0) -- (Q3) ;
             \draw (P1) -- (Q2) ;
             \draw (P1) -- (Q3) ;
             \draw (P2) -- (Q3) ;
             \draw (P1) -- (P0) ;
             \draw (P3) -- (P2) ;

             \draw (P3) to[out=115,in=65] (P0) ;
             \draw (P3) to[out=130,in=50] (P1) ;
             \draw (P2) to[out=130,in=50] (P0) ;
\end{tikzpicture}
\end{minipage}
\begin{minipage}[h]{6cm}
\begin{tikzpicture}

        \path (0,0) coordinate (P0);
        \path (1,0) coordinate (P1);
        \path (3,0) coordinate (P2);
        \path (4,0) coordinate (P3);
        \path (0,-1.5) coordinate (Q0);
        \path (1,-1.5) coordinate (Q1);
        \path (3,-1.5) coordinate (Q2);
        \path (4,-1.5) coordinate (Q3);
        \path (5,-0.75) coordinate (R);

        \foreach \x in {0,...,3} {
	\fill (P\x) circle (2pt);
	\fill (Q\x) circle (2pt);}
	\fill (R) circle (2pt);
        \foreach \x in {0,...,3} {
             \draw[dashed] (P\x) -- (Q\x) ;}

\path   (2,0) node {$\ldots$};
\path   (2,-1.5) node {$\ldots$};

             \draw (P0) -- (Q1) ;
             \draw (P0) -- (Q2) ;
             \draw (P0) -- (Q3) ;
             \draw (P1) -- (Q2) ;
             \draw (P1) -- (Q3) ;
             \draw (P2) -- (Q3) ;
             \draw (P1) -- (P0) ;
             \draw (P3) -- (P2) ;
             \draw (P0) -- (R) ;
             \draw (P1) -- (R) ;
             \draw (P2) -- (R) ;
             \draw (P3) -- (R) ;

             \draw (P3) to[out=115,in=65] (P0) ;
             \draw (P3) to[out=130,in=50] (P1) ;
             \draw (P2) to[out=130,in=50] (P0) ;
\end{tikzpicture}
\end{minipage}
\vspace{-0.1cm}
\caption{Graphs with maximal total irregularity $H_n$ (with dashed edges) and $\overline{H}_n$
               (without dashed edges) for even and odd $n$, respectively.\label{max-total-irr}} 
\end{figure}
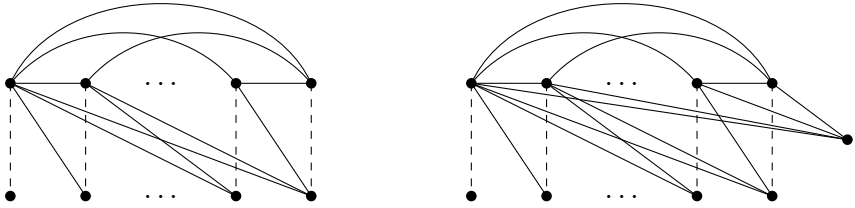
\vspace{-0.1cm}
The motivation to introduce the total irregularity of a graph, as modification of the  irregularity of graph, is twofold.
First, in contrast to $\irr(G)$, $\irr_t(G)$ can be computed directly from the sequence of the vertex degrees (degree sequence) of $G$.
Second, the most irregular graphs with respect to  $\irr$ are graphs that have only two degrees (see Figure~\ref{KS-max-irr}  for an illustration).
On the contrary the most irregular graphs with respect to  $\irr_t$, as it is shown in~\cite{Dimit-Abdo}, are graphs with maximal number of 
different vertex degrees (graphs with all doted (optional) edges in  Figure~\ref{max-total-irr}), which is much closer to what one can expect 
from ``very" irregular  graphs.

The aim of this paper is to investigate the total irregularity of graphs under several graph operations including   join, Cartesian product, 
direct product, strong product, lexicographic product, corona product, disjunction  and symmetric difference. Detailed exposition on some
graph operations one can find  in \cite{Imrich-Klavzar}.
%
\section[Results]{Results}
%
We start with simple observations about the complement and the disjoint union.

The {\it complement} of a simple graph $G$ with $n$ vertices, denoted by $\overline{G}$, is a simple
graph with $V(\overline{G})=V(G)$ and 
$E(\overline{G})=\{uv \, | \, u, v \in V(G) \; {\rm and} \; uv \notin E(G)\}$.
Thus, $uv \in E(G) \Longleftrightarrow uv \notin E(G)$. Obviously,
$E(G) \cup E(\overline{G})= E(K_n)$, and for a vertex $u$, we have
$d_{\overline{G}}(u)= n-1 - d_{G}(u)$. 
From $|d_{\overline{G}}(u)-d_{\overline{G}}(v)|= |n-1 - d_{G}(u) -(n-1 - d_{G}(v))|=
|d_{G}(u)-d_{G}(v)|$ it follows that
$
\irr_t(\overline{G}) = \irr_t(G).
$

For two graphs $G_1$ and $G_2$ with disjoint vertex sets $V(G_1)$
and $V(G_2)$ and disjoint edge sets $E(G_1)$ and $E(G_2)$ the {\it disjoint union} of $G_1$ and
$G_2$ is the graph $G = G_1 \cup G_2$ with the vertex set $V(G_1) \cup V(G_2)$ and the edge set
$E(G_1) \cup E(G_2)$. Obviously,
$
\irr_t(G \cup H) \geq \irr_t(G) + \irr_t(H).
$

Next we present sharp upper bounds for join, lexicographic product, Cartesian product,
strong product, direct product, corona product and upper bounds for disjunction and symmetric difference.
%
\subsection[Join]{Join}
%
The {\em join} $G + H$ of simple undirected graphs $G$ and $H$ is the graph with
the vertex set $V(G + H)=V (G)\cup V (H)$
and the edge set $E(G + H)=E(G)\cup E(H)\cup \left\{uv : u \in V(G), \; v \in V(H) \right\}$.
\begin{te} \label{irr2_Join}
Let $G$ and $H$ be simple undirected graphs with $\left|V(G)\right|= n_1$ and
$\left|V(H)\right| = n_2$ such that $n_1 \geq n_2$. 
Then,
\beq 
\irr_t(G + H) 
&\leq& \irr_t(G) + \irr_t(H) +  n_2 \, ( n_1 - 1)\,( n_1 - 2).  \nonumber 
\eeq	
Moreover, the bound is best possible.																
\end{te}
%
%
\begin{proof} 
The total irregularity of $G + H$ is
\beq \label{eqn-irr2_Join:00}
\irr_t(G + H) 
   &=& \frac{1}{2}\, \sum\limits_{u , v \in V(G + H)} \left|d_{G+H}(u)-d_{G+H}(v) \right| \nonumber \\
   &=& \frac{1}{2}\, \sum\limits_{u,v \in V(G)} \left| d_{G+H}(u)-d_{G+H}(v) \right| \nonumber \\
& \;\;\;+ &\frac{1}{2}\, \sum\limits_{u,v \in V(H)} \left| d_{G+H}(u)-d_{G+H}(v) \right| \nonumber \\
	 & \;\;\;+ &  \; \; \sum\limits_{u\in V(G)}\sum\limits_{v\in V(H)} \left| d_{G+H}(u)-d_{G+H}(v) \right|. \nonumber
\eeq
By definition, $\left| V(G+H) \right|= \left|V(G)\right|+\left|V(H)\right|=n_1 + n_2$. For vertices 
$u \in V(G)$ and $v \in V(H)$, it holds that $d_{G+H}(u) = d_G(u) +n_2$ and $d_{G+H}(v) = d_H(v) + n_1$. 
Thus, further we have
\beq \label{eqn-irr2_Join:01}
\irr_t(G + H)  
 &=&  \frac{1}{2}\, \sum\limits_{u,v \in V(G)}  \left| d_G(u)-d_G(v) \right|
		   + \frac{1}{2}\, \sum\limits_{u,v \in V(H)} \left| d_H(u)-d_H(v) \right| \nonumber \\
& & + \, \,\sum\limits_{u\in V(G)}\sum\limits_{v\in V(H)}\left|(d_G(u)+n_2)-(d_H(v)+n_1) \right| \\
&=&	\irr_t(G) + \irr_t(H) + \,\sum\limits_{u\in V(G)}\sum\limits_{v\in V(H)} \left| n_1 - n_2 + d_H(v) - d_G(u) \right|. \nonumber
\eeq
\noindent
Under the constrains $n_1 \geq n_2$, $d_G(u)  \leq n_1 -1$, and $d_H(v)  \leq n_2 -1$,
the double sum $\sum\limits_{u\in V(G)}\sum\limits_{v\in V(H)} \left| n_1 - n_2 + d_H(v) - d_G(u) \right|$ is maximal when
$H$ is a graph with maximal sum of vertex degrees, i.e., $H$ is the complete graph $K_{n_2}$, and
$G$ is a graph with minimal sum of vertex degrees, i.e., $G$ is a tree on  ${n_1}$ vertices $T_{n_1}$.
Thus,
\begin{eqnarray}  \label{eqn-M3_Join:01_}
\lefteqn{\sum\limits_{u\in V(G)} \sum\limits_{v\in V(H)} \left| n_1 - n_2 + d_H(v) - d_G(u) \right|}   
\qquad \qquad  \qquad \qquad  \qquad  \nonumber \\ 
 & \leq & \sum\limits_{u\in V(T_{n_1})}\sum\limits_{v\in V(K_{n_2})} \left| n_1 - n_2 +d_{K_{n_2}}(v)-d_{T_{n_1}}(u) \right|   \nonumber \\ 
 &=&\sum\limits_{u\in V(T_{n_1})}\sum\limits_{v\in V(K_{n_2})} \left| n_1 - 1 - d_{T_{n_1}}(u)  \right| \nonumber \\ 
 &=& n_2 \sum\limits_{u\in V(T_{n_1})} \left( n_1-1-d_{T_{n_1}}(u)  \right) \nonumber \\ 
 &=& n_2  n_1 (n_1 -1) - 2 n_2 (n_1- 1 ) \nonumber \\ 
 &=& n_2  (n_1 -1) (n_1 - 2 ) , \nonumber
\end{eqnarray}  
and 
\beq  \label{eqn-M3_Join:02}
\irr_t(G + H)  \leq    \irr_t(G) + \irr_t(H) +  n_2 ( n_1 - 1 ) ( n_1 - 2). 
\eeq
\noindent
When $n_1 \leq 2$, $\irr_t(G)= \irr_t(H)=\irr_t(G + H)=0$, and the claim of the theorem is fulfilled. 
From the derivation, it follows that  (\ref{eqn-M3_Join:02}) is equality  when 
$H$ is compete graph on $n_2$ vertices and $G$ is any tree on $n_1$ vertices . 
\end{proof}
%
%
\noindent
{\bf Example.} 
Let denote by $H_i$ a graph with $|V(H_i)|=i$ isolated vertices (vertices with degree zero).
Then, the {\it bipartite graph} $K_{i,j}$ is a join of $H_{i}$ and $H_{j}$.
Analogously, the {\it complete $k$-partite graph} $G=K_{n_1,\cdots, n_k}$
is join of $H_{n_1}, \dots, H_{n_k}$.
Straightforward calculation shows that
$\irr_t(K_{n_i,n_j}) =$ $ n_i n_j$ $\left| n_j - n_i \right|$,
For the total irregularity of $K_{n_1,\cdots, n_k}$ we have
\begin{align}  
\irr_t(K_{n_1,\cdots, n_k}) &= 
             \frac{1}{2}\,\sum\limits_{u,v\in V(K_{n_1,\cdots, n_k})} \left| d_G(u) - d_G(v) \right| \nonumber\\
       &= \sum\limits^{k-1}_{i=1}  \sum\limits^{k}_{j=i+1}  \left(\frac{1}{2}\,\sum\limits_{u,v\in V(K_{n_i,n_j})} 
				  \left| d_{G}(u)-d_{G}(v) \right| \right) \nonumber \\			
       &= \sum\limits^{k-1}_{i=1}  \sum\limits^{k}_{j=i+1} n_i n_j \left| n_j - n_i\right| 
	     = \sum\limits^{k-1}_{i=1}  \sum\limits^{k}_{j=i+1} \irr_t(K_{n_i,n_j}). \nonumber 
\end{align}
%
\subsection[lexicographic product]{Lexicographic product}
%
The {\em lexicographic product} $ G \circ  H $ (also known as the {\em graph composition}) 
of simple undirected graphs $G$ and $H$ is the graph with the vertex set $V( G \circ  H ) = V(G) \times V(H)$
and the edge set $E( G \circ  H ) = \{(u_i,v_k)(u_j,v_l) : [u_iu_j$ $\in E(G)]\vee [(v_k v_l \in E(H)) \wedge (u_i=u_j) ]\}$.
\begin{te} \label{irr2_compos}
Let $G$ and $H$ be simple undirected graphs with $\left|V(G)\right|= n_1$,
$\left|V(H)\right|= n_2$ then,
$$ 
\irr_t( G \circ  H ) \leq  \, n^3_2 \; \irr_t(G) +  n^2_1 \; \irr_t(H). 
$$ 
Moreover, this bound is sharp for infinitely many graphs.
\end{te}
\begin{proof}  
By the definition of $ G \circ  H $, it follows that
$\left|V( G \circ  H )\right|= n_1 n_2 $ and $d_{ G \circ  H }(u_i,v_j) = n_2 d_G(u_i) + d_H(v_j)$ 
for all $1 \leq i\leq n_1 , 1 \leq j\leq n_2$. Applying those relations, we obtain

\beq \label{eqn-irr2_comp:00}
\irr_t( G \circ  H ) 
&=& \frac{1}{2}\, \sum_{\substack{(u_i,v_k)\in V( G \circ  H )\\(u_j,v_l)\in V( G \circ  H )}}
       \left| d_{ G \circ  H }(u_i,v_k) - d_{ G \circ  H }(u_j,v_l) \right| \nonumber \\	
&=& \frac{1}{2}\, \sum_{\substack{u_i, u_j\in V(G) \\ v_k, v_l\in V(H)}}
		 \left| n_2 d_G(u_i) - n_2 d_G(u_j) 	+	d_H(v_k) - d_H(v_l) \right| \nonumber \\
&\leq&  \frac {1}{2}\, \sum_{\substack{u_i, u_j\in V(G) \\ v_k, v_l\in V(H)}}
         \left( n_2 \left| d_G(u_i)- d_G(u_j) \right| + \left| d_H(v_k) - d_H(v_l) \right| \right) \nonumber \\
&=&  \; \frac{1}{2} \; n^3_2 \;  \sum\limits_{u_i, u_j\in V(G)} \left| d_G(u_i)- d_G(u_j) \right|  \nonumber \\
&  & +  \;\frac{1}{2} \;  n^2_1 \; \sum\limits_{v_k, v_l\in V(H)} \left| d_H(v_k) - d_H(v_l) \right| \nonumber \\			
&=& n^3_2 \; \irr_t(G) +  n^2_1 \;\irr_t(H).
\eeq
To prove that the presented bound is best possible,
consider the lexicographic product $P_l \circ  C_k$, $l \geq 1, k \geq 3$ (an illustration
is given in Figure~\ref{AllOperations}(b)). Straightforward calculations give that
$ \irr_t(P_l)  = 2(l-2)$, $\irr_t(C_k) = 0$.
The graph $P_l \circ  C_k$ is comprised of $2 k$ vertices of degree $k+2$, and $k(l-2)$ vertices
of degree $2k+2$. Hence, $ \irr_t(P_l \circ  C_k) = 2 k^3 (l-2)$. 
On the other hand, the bound obtain here is
$\irr_t(P_l \circ  C_k) \leq  k^3 \; \irr_t(P_l) +  l^2 \;\irr_t(C_k) = 2 k^3 (l-2)$.
\end{proof}
%
\begin{figure*}[htbp]
\begin{center}
\includegraphics{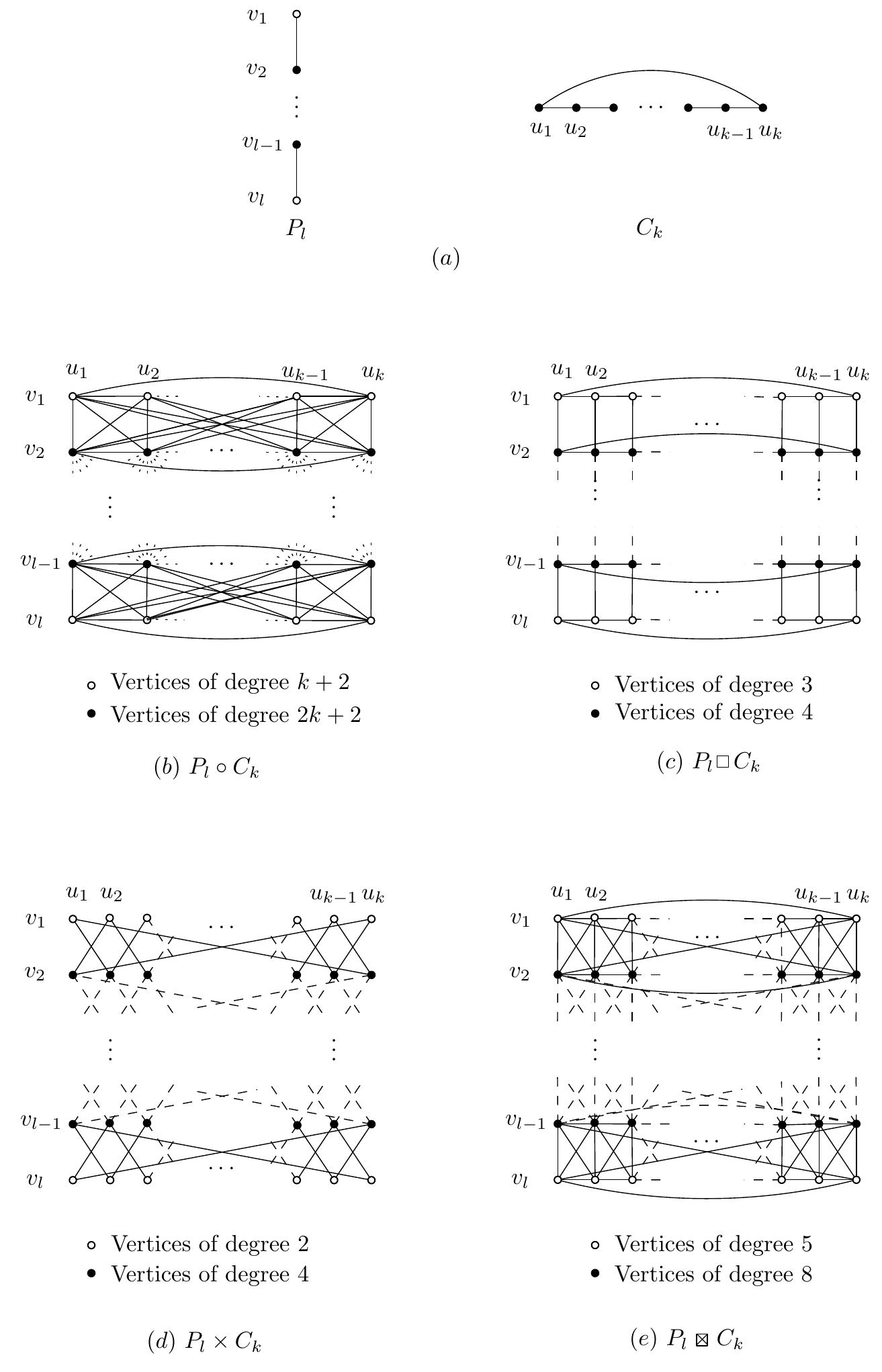}
\end{center}
\caption{
	       $(a)$ Path graph on $l$ vertices $P_l$, and cycle graph on $k$ vertices $C_k$,\,
		   $(b)$ lexicographic product graph $P_l \circ C_k$,\,
           $(c)$ Cartesian product graph $P_l \,\Box\, C_k$,\,
           $(d)$ direct product graph $P_l \times C_k$ \,and,\,
           $(e)$ strong product graph $P_l \boxtimes C_k$.
         }
\label{AllOperations}
\end{figure*}
%
\subsection[Cartesian product]{Cartesian product}

The {\em Cartesian product} $G\, \Box\, H$ of two simple undirected graphs $G$ 
and $H$ is the graph with the vertex set $V(G \,\Box\, H) = V(G) \times V(H)$ and the 
edge set
$E(G\, \Box\, H) = \{(u_i,v_k)(u_j,v_l): [(u_iu_j \in E(G))\wedge$$(v_k = v_l)]
\vee [(v_k v_l \in E(H))\wedge (u_i = u_j)]\}$.
From the definition of the Cartesian product, it follows that
$\left|V(G\, \Box\, H)\right|= n_1 n_2 $ and $d_{G\, \Box\, H}(u_i,v_j) = d_G(u_i) + d_H(v_j)$. 
Since the derivation of the upper bound on $G\, \Box\, H$ is similar to
the case of a graph lexicographic product, we omit the proof and just state the result in 
Theorem~\ref{irr2_Cartesian_pro}.
The best possible bound is obtained for $P_l\,\Box\, C_k$, $l \geq 1, k \geq 3$, illustrated in Figure~\ref{AllOperations}(c).
The graph $P_l\,\Box\, C_k$ is comprised of $2k$ vertices of degree $3$, and $k(l-2)$ vertices of degree $4$.
Thus, $ \irr_t(P_l\,\Box\, C_k) = 2 k^2 (l-2)$. The bound obtain here is
$\irr_t(P_l\,\Box\, C_k) \leq  k^2 \; \irr_t(P_l) +  l^2 \;\irr_t(C_k) = 2 k^2 (l-2)$.
\begin{te} \label{irr2_Cartesian_pro}
Let $G$ and $H$ be simple undirected graphs with $\left|V(G)\right|= n_1$, $\left|V(H)\right| = n_2$ then
$$ \irr_t(G\, \Box\, H) \leq n_2^2\; \irr_t(G) + n_1^2\; \irr_t(H).$$
Moreover, this bound is sharp for infinitely many graphs.
\end{te}
%
%
\subsection[strong product]{Strong product}
%
The {\em strong product} $G \boxtimes H$ of two simple undirected graphs $G$ 
and $H$ is the graph with the vertex set $V(G \boxtimes H) = V(G) \times V(H)$ and the 
edge set
$E(G \boxtimes H) = \{(u_i,v_k)(u_j,v_l): [(u_iu_j \in E(G))\wedge$ $(v_k = v_l)]
\vee [(v_k v_l \in E(H))\wedge (u_i = u_j)]$ $\vee$$ [(u_iu_j \in E(G)) \wedge (v_k v_l \in E(H))]\}$. 
\begin{te} \label{irr2_Strong_pro}
Let $G$ and $H$ be simple undirected graphs with $\left|V(G)\right|= n_1$, 
$\left|E(G)\right|= m_1$, $\left|V(H)\right|= n_2$ and  $\left|E(H)\right|= m_2$.
Then,
$$ \irr_t(G\boxtimes H) \leq \, n_2 \, (n_2 + 2 m_2)\; \irr_t(G) + n_1 \, (n_1 + 2 m_1)\; \irr_t(H).$$
Moreover, this bound is best possible.
\end{te}
\begin{proof}
From the definiton of the strong product, it follows 
$\left| V(G \boxtimes H) \right| = n_1 n_2$, 
$\left|E(G \boxtimes H)\right| =  m_1 n_2 + m_2 n_1 + 2 m_1 m_2$,
and  $d_{G \boxtimes H}(u_i,v_k)=d_G (u_i) + d_H (v_k) + d_G (u_i) d_H (v_k)$.
The total irregularity of $G \boxtimes H$ is
\beq \label{eqn-irr2_str-pro:00}
\irr_t(G \boxtimes H) 
&=& \frac{1}{2} \sum\limits_{(u_i,v_k),(u_j,v_l)\in V(G\boxtimes H)} \left| d_{G \boxtimes H}(u_i,v_k) -
                                              d_{G \boxtimes H}(u_j,v_l) \right| \nonumber \\
&=& \frac{1}{2} \sum\limits_{u_i, u_j \in V(G), v_k, v_l \in V(H)} 
	    \left| d_{G \boxtimes H}(u_i,v_k) - d_{G \boxtimes H}(u_j,v_l) \right|.
\eeq	
Applying simple algebraic transformation and the triangle inequality, we obtain
\beq \label{eqn-irr2_str-pro:01}
\left| d_{G \boxtimes H}(u_i,v_k) - d_{G \boxtimes H}(u_j,v_l) \right| 
&=&
    \left| ( d_G(u_i)-d_G(u_j) ) + (d_H(v_k) -d_H(v_l))  \right. \qquad \nonumber \\
& & \left. + \, ( d_G(u_i) d_H(v_k)- d_G(u_j) d_H(v_l)) \right| \nonumber \\
&\leq& 
      \left| d_G(u_i)-d_G(u_j) \right| + \left| d_H(v_k) -d_H(v_l)\right| \nonumber \\
&+& \frac{1}{2} (d_G(u_i)+d_G(u_j)) \left| d_H(v_k)-d_H(v_l) \right| \nonumber \\
&+& \frac{1}{2} (d_H(v_k)+d_H(v_l)) \left| d_G(u_i)-d_G(u_j) \right|.   
\eeq
From (\ref{eqn-irr2_str-pro:00}) and (\ref{eqn-irr2_str-pro:01}), we obtain
\beq \label{eqn-irr2_str-pro:03}
\irr_t(G\boxtimes H)        
              &\leq&  \frac{1}{2} \sum_{\substack{u_i, u_j \in V(G),\\ v_k, v_l \in V(H)}}
						    \left[\left| d_G(u_i)-d_G(u_j) \right| \right.
						+  \left. \left| d_H(v_l)-d_H(v_k) \right| \right]	\nonumber \\
		      &   &  + \frac{1}{4} \sum_{\substack{u_i, u_j \in V(G),\\ v_k, v_l \in V(H)}}
									   (d_G(u_i)+d_G(u_j)) \left| d_H(v_k)-d_H(v_l) \right|  \nonumber \\
             &   &  + \frac{1}{4} \sum_{\substack{u_i, u_j \in V(G),\\ v_k, v_l \in V(H)}}
									   (d_H(v_k)+d_H(v_l)) \left| d_G(u_i)-d_G(u_j) \right|	 \nonumber \\												
     		 & = &  n_2 \, (n_2 + 2 m_2)\; \irr_t(G) + n_1 \, (n_1 + 2 m_1)\; \irr_t(H).	\nonumber
\eeq	
To prove that the presented bound is best possible, consider the strong product $P_l \otimes C_k$, 
$l \geq 1, k \geq 3$,
illustrated in Figure~\ref{AllOperations}(e). 
We have,
$\irr_t(P_l) = 2(l-2)$, $\irr_t(C_k) = 0$. 
The graph $P_l\boxtimes C_k$ is comprised of $2 k$ vertices of degree $5$, and $k(l-2)$ vertices
of degree $8$. Hence, $ \irr_t(P_l\boxtimes C_k) = 6 k^2(l-2)$. 
On the other hand, the bound obtain here, is
$ \irr_t(P_l\boxtimes C_k) \leq  l (l + 2 (l-1))\; \irr_t(C_k)  +  k (k + 2 k)\; \irr_t(P_l)
= 6 k^2 (l-2)$.

\end{proof}
%
\subsection[direct product]{Direct product}
%
The {\em direct product} $G \times H$ (also know as the {\em tensor product}, the 
{\em Kronecker product}~\cite{P-Weichsel}, 
{\em categorical product}~\cite{D-Miller} and {\em conjunctive product}) of simple 
undirected graphs $G$ and $H$ is the graph with the vertex set $V(G \times H)=V(G) \times V(H)$,
and the edge set $E(G\times H)=\{(u_i,v_k)(u_j,v_l):(u_i,u_j)\in\,$ $E(G)\wedge(v_k,v_l)\in E(H)\}$.
From the definition of the direct product, it follows 
$\left|V(G \times H) \right|= n_1 n_2$, 
$\left| E(G \times H) \right| = 2 m_1 m_2$, 
and  $d_{G \times H}(u_i,v_k)=d_G (u_i) d_H (v_k)$. 
The proof for the upper bound on $G \times H$ is similar as
that of the strong product $G \boxtimes H$. Therefore, we show only that the bound
in Theorem~\ref{irr2_Tensor_pro} is best possible, and omit the rest of the proof.

\begin{te} \label{irr2_Tensor_pro}  
Let $G$ and $H$ be simple undirected graphs with $\left|V(G)\right|= n_1$,
$\left|E(G)\right|= m_1$, $\left|V(H)\right|= n_2$ and $\left|E(H)\right|= m_2$.
Then,
$$ \irr_t(G \times H) \leq 2 \,n_2\,m_2\;\irr_t(G) + 2\,n_1 \,m_1 \;\irr_t(H).$$
Moreover, this bound is best possible.
\end{te}

To prove that the presented bound is best possible,
we consider the direct product $P_l \times C_k$, $l \geq 1, k \geq 3$
(an illustration is given in Figure~\ref{AllOperations}(d).
Straightforward calculations give that
$\irr_t(P_l) = 2(l-2)$, $ \irr_t(C_k) = 0$.
The graph $P_l\times C_k$ is comprised of $2 k$ vertices of degree $2$, and $k(l-2)$ vertices
of degree $4$. Thus, $ \irr_t(P_l\times C_k) = 4 k^2(l-2)$. 
On the other hand, the bound obtain by Proposition~\ref{irr2_Tensor_pro} is
$\irr_t(P_l\times C_k) \leq  2 n_2 m_2\,\irr_t(G) + 2 n_1 m_1 \, \irr_t(H) = 4 k^2 (l-2)$.
%
\subsection[Corona product]{Corona product}
The {\em corona} product $G \odot H$ of simple undirected graphs $G$ and $H$ with 
$\left|V(G)\right|=n_1$ and $\left|V(H)\right|=n_2$, is defined as the graph who is
obtained by taking the disjoint union of $G$ and $n_1$ copies of $H$ and for each 
$i$, $1\leq i \leq n_1$, inserting edges between the $i$th vertex of $G$ and each vertex
of the $i$th copy of $H$. Thus, the corona graph $G \odot H$ is the graph with the 
vertex set $V(G \odot H) = V(G)\cup_{i=1,\dots,n_1} V (H_i)$ and the edge set 
$E(G\odot H)=E(G)\cup_{i=1,\dots,n_1} E(H_i)\cup \left\{ u_iv_j:u_i\in V(G), v_j\in V(H_i)\right\}$, 
where $H_i$ is the $i$th copy of the graph $H$.
%
\begin{te} \label{irr2_corona}  
Let $G$ and $H$ be simple undirected graphs
with $\left| V(G) \right|=n_1$ and $\left| V(H)\right|=n_2$. Then,
\beq
\irr_t(G \odot H)  \leq \irr_t (G) + n^2_1 \irr_t (H) + n^2_1 \left( n^2_2 + n_1 n_2 - 4 n_2 + 2 \right). \nonumber
\eeq
Moreover, the bound is best possible.
\end{te}
\begin{proof} 
The total irregularity of $G \odot H$ is
\beq \label{eqn-irr2_corona:00}
\irr_t(G \odot H) 
&=&  \frac{1}{2}\, \sum\limits_{u, v \in V(G \odot H)}\left| d_{G \odot H}(u)-d_{G \odot H}(v) \right| \nonumber \\
&=&  \frac{1}{2}\, \sum\limits_{x,y\in V(G)}\left| d_{G \odot H}(x)-d_{G \odot H}(y) \right| \nonumber \\
&   & + \sum\limits^{n_1}_{i=1}\left( \frac{1}{2}\,\sum\limits_{z,t\in V(H_i)}
	     \left| d_{G \odot H}(z) - d_{G \odot H}(t) \right| \right) \nonumber \\
& & + \sum\limits^{n_1 -1}_{i=1}\sum\limits^{n_1}_{j=i+1} \, 
       \sum_{\substack{z\in V(H_i), t\in V(H_j)}}
        \left| d_{G \odot H}(z) - d_{G \odot H}(t) \right|  \nonumber \\
& & +    \sum\limits^{n_1}_{i=1}\, \sum\limits_{u\in V(G), v\in V(H)} 
	        \left| d_{G \odot H}(u)-d_{G \odot H}(v) \right|.	\nonumber
\eeq	
By the definition of $G \odot H$, 
$\left| V(G \odot H) \right| = \left| V(G) \right| + n_1 \left| V(H) \right|$ $=n_1+n_1 n_2$.
For a vetrex $u \in V(G)$, it holds that $d_{G \odot H}(u) = d_G(u) +n_2$ and 
for a vertex $v \in V(H_i)$, $1 \leq i \leq n_2$, we have $d_{G \odot H}(v) = d_H(v) + 1$. Thus,
\bigskip				
\beq	\label{eqn-irr2_corona:01}	
\irr_t(G \odot H) 
 &=& \frac{1}{2}\, \sum\limits_{x,y\in V(G)} \left| d_G(x)-d_G(y) \right| 
     +  \frac{1}{2} \, n_1 \, \sum\limits_{z,t\in V(H)} \left| d_H(z) - d_H(t) \right| \nonumber \\
 & & + \sum\limits^{n_1 -1}_{i=1}\sum\limits^{n_1}_{j=i+1} 
           \, \sum_{\substack{z\in V(H_i),\, t\in V(H_j)}}
		        \left| d_H(z)+1 - d_H(t)-1 \right| \nonumber \\
& & + \sum\limits^{n_1}_{i=1} \, \sum\limits_{u\in V(G), \, v\in V(H)}
       \left| d_G(u)-d_H(v) + n_2 - 1 \right| \nonumber \\	
&=& \irr_t (G) + n_1 \irr_t (H) +  n_1 ( n_1 - 1 ) \irr_t (H) \nonumber \\
 & &	+ \sum\limits^{n_1}_{i=1}
        \sum\limits_{u\in V(G),\, v\in V(H)} \left| d_G(u)-d_H(v) + n_2 - 1 \right|. 
\eeq
\noindent
Since $n_1 \geq n_2$,  the sum $\sum\limits_{u\in V(G),\, v\in V(H)} \left| d_G(u)-d_H(v) + n_2 - 1 \right|$
is maximal when 
$\sum_{u\in V(G)}  d_G(u) $ is maximal, i.e., $G$ is the complete graph $K_{n_1}$,  and $\sum_{v\in V(H)}  d_H(v) $ is minimal, i.e., 
$H$ is a tree on $n_2$ vertices $T_{n_2}$.
Thus,
\bigskip
%
\begin{eqnarray} \label{eqn-M3_corona:01_1}		
\lefteqn{ \sum\limits_{u\in V(G),\, v\in V(H)} \left| d_G(u)-d_H(v) + n_2 - 1 \right|} \nonumber\\
& \leq & \sum\limits_{u\in V(K_{n_1})}\sum\limits_{v\in V(T_{n_2})} \left| d_{K_{n_1}}(u)-d_{T_{n_2}}(v) + n_2 - 1 \right|  \nonumber\\
&=& \sum\limits_{u\in V(K_{n_1})}\sum\limits_{v\in V(T_{n_2})} \left| n_1 - 1 - d_{T_{n_2}}(v) + n_2 - 1 \right|   \nonumber\\	              	
&=&  n_1 \sum\limits_{v\in V(T_{n_2})} \left( n_1  + n_2 - 2 - d_{T_{n_2}}(v) \right)   \nonumber\\	
&=&  n_1 n_2 (n_1  + n_2 - 2 ) -  2 n_1( n_2 -1 )   \nonumber\\	
&=&  n_1 ( n_2^2 + n_1 n_2 - 4  n_2 + 2 ). 
\end{eqnarray} 
\smallskip
Substituting (\ref{eqn-M3_corona:01_1}) into  (\ref{eqn-irr2_corona:01}), we obtain 
%
%
\beq  \label{eqn-M3_corona:03}
\irr_t(G \odot H)  \leq \irr_t (G) + n^2_1 \irr_t (H) + n^2_1 \left( n^2_2 + n_1 n_2 - 4 n_2 + 2 \right).
\eeq	
From the derivation of the bound (\ref{eqn-M3_corona:03}), it follows that the sharp bound is obtained when 
$G$ is compete graph on $n_1$ vertices and $H$ is any tree on $n_2$ vertices. 			
\end{proof}
\subsection[disjunction]{Disjunction}
%
The {\em disjunction} graph $G \vee H$ of simple undirected graphs $G$ and $H$ with
$\left|V(G)\right|=n_1$ and $\left|V(H)\right|=n_2$ is the graph with the vertex set
$V(G \vee H)=V(G) \times V(H)$ and the edge set 
$E(G \vee H) =$ $\left\{(u_i,v_k)(u_j,v_l):\right.$ $u_i u_j \in E(G)$ $\left. \vee v_k v_l \in E(H)\right\}$.
It holds that $\left|V(G \vee H)\right|= n_1 n_2 $,
and $d_{G\vee H}(u_i,v_k) = n_2 d_G(u_i) + n_1 d_H(v_k) - d_G(u_i) d_H(v_k)$
for all $i$, $k$ where, $1 \leq i\leq n_1$, $1 \leq k\leq n_2$.
%
\begin{te} \label{irr2disj}
Let $G$ and $H$ be simple undirected graphs with $\left|V(G)\right|= n_1$, 
$\left|E(G)\right|= m_1$, $\left|V(H)\right|= n_2$ and  $\left|E(H)\right|= m_2$.
Then,
$$ \irr_t(G\vee H) \leq  \, n_2 \, (n^2_2 + 2 m_2)\; \irr_t(G) + n_1 \, (n^2_1 + 2 m_1)\; \irr_t(H). $$
\end{te}
\begin{proof}
The total irregularity of  $G\vee H$ is
\beq \label{eqn-irr2_disj:00}
\irr_t(G\vee H) 
         &=&  \frac{1}{2}\, \sum\limits_{(u_i,v_k),(u_j,v_l) \in V(G\vee H)} \left| d_{G\vee H}(u_i,v_k) -
                                                                        d_{G\vee H}(u_j,v_l) \right|   \nonumber \\
         &=&  \frac{1}{2}\, \sum\limits_{u_i, u_j \in V(G), v_k, v_l \in V(H)}
						                                              \left| d_{G\vee H}(u_i,v_k) - d_{G\vee H}(u_j,v_l) \right|.
\eeq	
\noindent
Since $d_{G\vee H}(u_i,v_k) = n_2 d_G(u_i) + n_1 d_H(v_k) - d_G(u_i) d_H(v_k)$
for all $i$, $k$ where, $1 \leq i\leq n_1$, $1 \leq k\leq n_2$. We obtain
\beq \label{eqn-irr2_disj:02-00}
\left| d_{G\vee H}(u_i,v_k) - d_{G\vee H}(u_j,v_l) \right|
   &=&  \left|n_2 d_G(u_i) + n_1 d_H(v_k) - d_G(u_i) d_H(v_k) \right. \nonumber \\
   & & - \left.(n_2 d_G(u_j) + n_1 d_H(v_l) - d_G(u_j) d_H(v_l)) \right| \nonumber
\eeq
\noindent
Further, by simple algebraic manipulation and by the triangle inequality, we have
\beq \label{eqn-irr2_disj:02}
\left| d_{G\vee H}(u_i,v_k) - d_{G\vee H}(u_j,v_l) \right|
&\leq& n_2 \left| d_G(u_i)-d_G(u_j) \right| + n_1 \left| d_H(v_k) - d_H(v_l)\right| \nonumber \\
&    &  +\frac{1}{2}  (d_G(u_i)+d_G(u_j)) \left| d_H(v_k)-d_H(v_l) \right| \nonumber \\
&    &  +\frac{1}{2}  (d_H(v_k)+d_H(v_l)) \left| d_G(u_i)-d_G(u_j) \right|.\nonumber \\
\eeq
From (\ref{eqn-irr2_disj:00}) and (\ref{eqn-irr2_disj:02}), we obtain

\begin{eqnarray}  \label{eqn-irr2_disj:03}
%
\irr_t(G\vee H)  &\leq&  \frac{1}{2} \sum_{\substack{u_i, u_j \in V(G), \\ v_k, v_l \in V(H)} }
							           \left[ n_2 \left| d_G(u_i)-d_G(u_j) \right|
								    +  n_1 \left| d_H(v_l)-d_H(v_k) \right|\right] 	\nonumber \\			
			&    & + \frac{1}{4} \sum_{\substack{u_i, u_j \in V(G), \\ v_k, v_l \in V(H)} }
							            (d_G(u_i)+d_G(u_j)) \left| d_H(v_k)-d_H(v_l) \right| \nonumber \\
			&    & +  \frac{1}{4} \sum_{\substack{u_i, u_j \in V(G), \\ v_k, v_l \in V(H)} }
                          (d_H(v_k)+d_H(v_l)) \left| d_G(u_i)-d_G(u_j) \right|. 
\end{eqnarray} 
\noindent
The first sum in (\ref{eqn-irr2_disj:03}) is equal to $n^3_2 \,\irr_t(G) + n^3_1 \,\irr_t(H)$,
the second to $2 n_1 m_1 \,\irr_t(H)$, and the third to $2 n_2 m_2 \,\irr_t(G)$. Hence,
\beq                      
 \irr_t(G\vee H) 
      &\leq&   n^3_2 \; \irr_t(G) + n^3_1 \; \irr_t(H) + 2\,  n_1 m_1 \;\irr_t(H) + 2 \, n_2 m_2 \; \irr_t(G) \nonumber \\		    
			&  = & n_2 \, (n^2_2 + 2 m_2)\; \irr_t(G) + n_1 \, (n^2_1 + 2 m_1)\; \irr_t(H).  \nonumber		
\eeq	
\end{proof}
\subsection[Symmetric difference]{Symmetric difference}
%
The {\em symmetric difference} $G \oplus H$ of simple undirected graphs
$G$ and $H$ with $\left|V(G)\right|=n_1$ and $\left|V(H)\right|=n_2$ is the 
graph with the vertex set $V(G \oplus H)=V(G) \times V(H)$ and the edge set 
$E(G \oplus H)=\left\{(u_i,v_k)(u_j,v_l): \, \mbox{either}\, u_iu_j \in E(G)
                                           \, \mbox{or}\, v_k v_l \in E(H) \right\}$.
It holds that
$\left|V(G \oplus H)\right|$ $= n_1 n_2 $,
and $d_{(G\oplus H)}(u_i,v_j) = n_2 d_G(u_i) + n_1 d_H(v_j) - 2 d_G(u_i) d_H(v_j)$
for all $1 \leq i\leq n_1 , 1 \leq j\leq n_2$.

Much as in the previous case, we present only the bound on the total irregularity of symmetric difference of two graphs.
\begin{te} \label{irr2_symm-diff}
Let $G$ and $H$ be simple undirected graphs with $\left|V(G)\right|= n_1$, 
$\left|E(G)\right|= m_1$, $\left|V(H)\right|= n_2$ and  $\left|E(H)\right|= m_2$.
Then,
$$ 
\irr_t(G\oplus H) \leq  n_2 \, (n^2_2 + 4 m_2)\; \irr_t(G) + n_1 \, (n^2_1 + 4 m_1)\; \irr_t(H). 
$$
\end{te}
\smallskip
\noindent
%
%
\section[Conclusion]{Conclusion}
In this paper we consider the total irregularity of simple undirected graphs under several graph operations.
We present sharp upper bounds for join, lexicographic product, Cartesian product, strong product, direct
product and corona product. It is an open problem if the presented upper bounds on the total irregularity 
of disjunction and symmetric difference are the best possible. 
%
%

%
%
\end{document}